\documentclass{article}
\usepackage{cite,comment}
\usepackage{amsmath,amssymb,amsfonts}
\usepackage{graphicx}
\usepackage{textcomp}
\usepackage{xcolor}
\usepackage{bbm}
\usepackage{amsthm,amsmath,fullpage,setspace,amsfonts,amsthm,placeins}
\usepackage{graphicx}
\usepackage{mathrsfs}
\usepackage{amsmath}
\usepackage{xcolor}
\usepackage{arydshln}
\usepackage{subfig}
\usepackage{algorithm}

\usepackage[noend]{algpseudocode}

\newcommand{\bF}{\mathbb{F}}

\newcommand{\cC}{\mathcal{C}}
\newcommand{\sC}{\mathscr{C}}

\newcommand{\supp}{\text{supp}}
\newcommand{\probset}{(R,B)}

\newcommand{\task}{T = (k,M)}
\newcommand{\model}{EIC}
\newcommand{\Gi}{G|_{\tilde{N}_i}}
\newcommand{\Cp}{(C)_{\probset}}
\newcommand{\Dp}{(D)_{\probset}}
\newcommand{\Tp}{(T)_{\probset}}
\newcommand{\Aprob}{\mathcal{A}_{\probset}}
\newcommand{\C}{(C)}
\newcommand{\D}{(D)}
\newcommand{\T}{(T)}
\newcommand{\N}{\tilde{N}}
\newcommand{\data}{\mathcal{D}}
\newcommand{\decvec}{\boldsymbol{\alpha}}
\newcommand{\indvec}{\boldsymbol{e}}
\newcommand{\sendmat}{\beta}
\DeclareMathOperator{\minrk}{minrk}
\DeclareMathOperator{\rminrk}{r-minrk}
\DeclareMathOperator{\rk}{rk}
\newtheorem{theorem}{Theorem}
\newtheorem{lemma}{Lemma}
\newtheorem{cor}{Corollary}
\newtheorem{definition}{Definition}

\newtheorem{remark}{Remark}
\newtheorem{algo}{Algorithm}
\newcommand{\mkw}[1]{\textbf{ \textcolor{blue}{[ #1 --mary ]} } }
\newcommand{\amp}[1]{\textbf{ \textcolor{purple}{[ #1 --alex ]} } }

\begin{document}



\title{Embedded Index Coding\thanks{A preliminary version of this work appeared at ITW 2019.}}
%
%
%
\author{Alexandra Porter and Mary Wootters\thanks{AP is with the Department of Computer Science, Stanford University.  MW is with the Departments of Computer Science and Electrical Engineering, Stanford University.  This work is partially supported by NSF grant CCF-1657049 and NSF CAREER grant CCF-1844628. AP is partially supported by the National Science Foundation Graduate Research Fellowship under Grant No. DGE-1656518}}

\markboth{Journal }%
{Porter and Wootters: Embedded Index Coding}
%



\maketitle

\begin{abstract}
Motivated by applications in distributed storage and distributed computation, we introduce \emph{embedded index coding} (EIC).  EIC is a type of distributed index coding in which nodes in a distributed system act as both broadcast senders and receivers of information. We show how embedded index coding is related to index coding in general, and give characterizations and bounds on the communication costs of optimal embedded index codes. We also define \emph{task-based}  EIC, in which each sending node encodes and sends data blocks independently of the other nodes. Task-based EIC is more computationally tractable and has advantages in applications such as distributed storage, in which senders may complete their broadcasts at different times. Finally, we give heuristic algorithms for approximating optimal embedded index codes, and demonstrate empirically that these algorithms perform well.

\end{abstract}


\section{Introduction}
\subsection{Motivation}

In 
\emph{index coding}, defined by Bar-Yossef, Birk, Jayram and Kol in~\cite{bar2006index}, sender(s) encode data blocks into messages which are broadcast to receivers.  The receivers already have some of the data blocks, and the goal is to take advantage of this ``side information'' in order to minimize the number of messages broadcast.
For example, if node $r_1$ knows a data block  $b_1$ and node $r_2$  knows block  $b_2$, a sender $S$ can broadcast $b_1\oplus b_2$. Then $r_1$ can cancel out $b_1$ and $r_2$ can cancel $b_2$ such that both nodes learn a distinct new block from a single broadcast message.

Index coding is typically studied in the models depicted in Figures~\ref{fig:icexample1} and \ref{fig:icexample2}, where the senders are distinct from the receivers.  In this paper, we consider a setting---depicted in Figure~\ref{fig:icexample3}---where the senders \em are \em the receivers.  This is similar to a ``peer-to-peer'' network model, but in this setting nodes are always communicating by broadcasting to the full network, rather than communicating with each other directly.  This model is motivated by applications in distributed storage and distributed computation.  For example, in coded computation, e.g.~\cite{li2016fundamental}, the \emph{shuffle} phase consists of nodes communicating computed values with each other. 
 
\begin{figure}
	\centering  
	\subfloat[{\label{fig:icexample1}}]{\includegraphics[width=.15\textwidth]{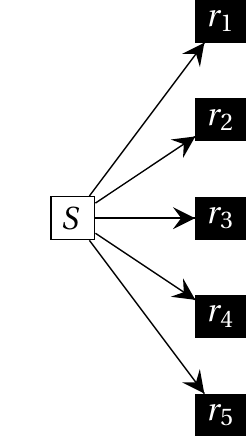}}
		\subfloat[{\label{fig:icexample2}}]{\includegraphics[width=.15\textwidth]{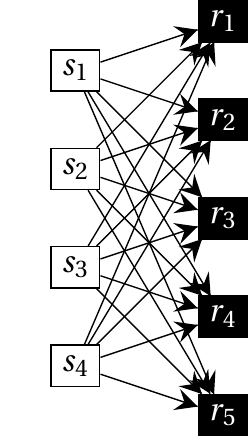}}
	\subfloat[{\label{fig:icexample3}}]{\includegraphics[width=.18\textwidth]{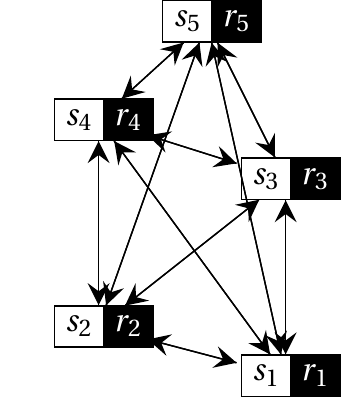}}
	\caption{Communication model for (a) centralized index coding with sender $S$, receivers $r_1,...,r_5$; (b) general multi-sender index coding with  senders $s_1,..,s_4$ and receivers $r_1,...,r_5$; and (c) embedded index coding, a special case of (b) with joint sender and receiver nodes $r_1=s_1,...,r_5=s_5$. 
	}\label{fig:modeltypes} 

\end{figure}

We call this model \em embedding index coding \em (EIC). EIC can be seen as a special case of the multi-sender index coding model in Figure~\ref{fig:icexample2}. In this paper, we will demonstrate that by considering EIC as a special case, we can prove new results and design faster algorithms than are available for the more general multi-sender index coding problem. 

We also introduce a new notion of solution to an embedded index coding problem called a \em task-based \em solution.  
In a task-based solutions, the communication can be partitioned into independent tasks, so that each receiver is only reliant on a single sender to get a particular block.  This can be seen as a generalization of \em
Instantly Decodable Network Codes\em~\cite{le2013instantly} which have been studied with similar motivation (see Remark~\ref{rem:icandidnc}).  Task based solutions are also related to \em Locally Decodable Index Codes\em~
\cite{natarajan2018locally} (see Remark~\ref{rem:ldic}). 

As we will see, there are efficient heuristics to find good task-based solutions to EIC problems.  Moreover, task-based solutions can be more robust to failures or delays: if a sender's messages are corrupted or lost, the messages from other senders can still be used fully to decode data blocks.

\subsection{Outline and Contributions}
In Section~\ref{sec:related} we review related  work in more detail. In Section~\ref{sec:framework} we formally define the EIC problem and several notions of solution. In Section~\ref{sec:relationships} we show how  EIC problems relate to more general index coding and we analyze how different notions of solutions are related. In Section~\ref{sec:algos} we provide algorithms for approximating optimal EIC solutions and demonstrate empirically that they perform well.

 Our contributions can be summarized as follows.
\begin{enumerate}
\item We define \emph{embedded index codes}, a type of distributed index code in which nodes function as both broadcast senders and receivers.
\item We define \emph{task-based index coding}, which seems more computationally tractable than a general solution to an EIC problem, and can be thought of as relaxing the concept of \emph{instantly decodability} in network codes. 
\item We prove several results establishing relationships between centralized (single-sender) index coding, EIC, and task-based EIC. In particular, we show that the optimal communication for a general EIC problem is only a factor of two worse than the optimal communication in the centralized model; we give characterizations and bounds for the optimal communication cost of the best task-based solutions to an EIC problem; and we show separations between the three models.
\item Based on the (proofs of) the bounds described above, we design heuristics for designing general EICs and and task-based EICs.  We give empirical evidence that these approximation algorithms perform well. 
\end{enumerate}

\section{Related Work}\label{sec:related}
In this section we briefly review related work.
Index coding was first introduced by~\cite{bar2006index}, based on the Informed-Source Coding on Demand (ISCOD) model proposed by~\cite{birk2006coding}, and many extensions and variations have been studied, including non-linear index coding~\cite{lubetzky2009nonlinear} and multi-sender index coding~\cite{ong2016single}. 
We focus on \em linear \em index coding, where the messages broadcast are linear combinations of the original data. 

The work of~\cite{bar2006index} 
characterized the number of broadcasts required to solve an index coding problem in terms of the \em minrank \em (c.f. Definition~\ref{def:minrank}) of a relevant graph. 
The minrank is difficult to compute exactly, and a number of approximations and heuristics have been studied for computing optimal linear index  codes~\cite{chaudhry2008efficient,chaudhry2011complementary,shanmugam2013local,neely2013dynamic,tahmasbi2015critical,thapa2017interlinked}.
We will also use the minrank, and heuristics for computing it, in our approach. 
%
%
 

Embedded index codes are a special case of the linear multi-sender index codes in~\cite{kim2019linear} and~\cite{li2018multi}, which both consist of multiple senders and multiple receivers, but as two distinct and non-overlapping sets of nodes; this is the setting depicted in Figure~\ref{fig:icexample2}.  In~\cite{li2018multi}, rank minimization is used in an approach similar to our method. The approaches of \cite{li2018multi,kim2019linear} can also be applied to EIC, and 
we compare these approaches in more detail in Section~\ref{sec:algos}.

The embedded model in Figure~\ref{fig:icexample3} has been studied before in \cite{el2010coding}.  In that work, the authors study a special case of EIC, where each node wants all of the data blocks it does not already have.  In this setting, they develop a greedy algorithm which uses a near-optimal number of broadcasts.  However, their approach crucially uses the fact that every node wants every block, and does not seem to generalize to the general EIC setting that we study here.

While our coding scheme is deterministic, our multi-sender network model is similar to those studied with \em composite coding, \em an approach based on randomized coding~\cite{arbabjolfaei2013capacity}. Multi-sender models and achievable rate regions using composite coding are defined in~\cite{sadeghi2016distributed,li2017improved,li2018cooperative}; to the best of our knowledge these results are not directly applicable to our scheme.

Index coding is a special instance of the network coding problem (e.g.,~\cite{li2003linear}), in which source nodes send information over a network containing intermediate nodes, which may modify messages, in addition to receiver nodes. It has also been shown that network coding instances can be reduced to index coding instances~\cite{el2008relation,effros2015equivalence}. Real-Time Instantly Decodable Network Codes (IDNC's)~\cite{le2013instantly} aim to minimize completion delay of the communication task, rather than the index coding goal of minimizing total number of messages. Our task-based solutions are a generalization of instant decodability in index codes (see Remark~\ref{rem:icandidnc}). 

Task-based solutions are also related to the notion of \em locally decodable index codes\em. An index coding solution has \em locality  \em $r$ if each node uses at most $r$ received symbols to decode any message symbol. There is tradeoff between optimal broadcast rate and locality of solutions for a given index coding problem~\cite{natarajan2018locally}. 
When $r=1$, locally decodable index codes are a special case of task-based schemes, although the notions diverge for more general $r$ (see Remark~\ref{rem:ldic}).

 Our construction is motivated by the problem of data shuffling for coded computation, such as in~\cite{li2016fundamental}, or for distributed storage systems which need to redistribute data among the nodes. 
  In data shuffling, after an initial round of computation, nodes each contain some amount of intermediate results, which then need to be shared with other nodes to continue the computation.
Other connections between index coding and distributed storage have been established,  but are not directly related to  our  work. These include the relationship between an optimal recoverable distributed storage code and a  general optimal index code~\cite{mazumdar2014duality} and the duality of linear index codes and Generalized Locally Repairable codes was shown by~\cite{shanmugam2014bounding,arbabjolfaei2015three}. 

Finally,  index coding techniques can also be applied to coded caching (e.g.~\cite{maddah2014fundamental},~\cite{ghasemi2017improved} and references therein), in  which nodes may request and store data dynamically. Coded multicasting similar to index coding has been applied to decentralized coded caching~\cite{maddah2015decentralized}, and our work could also be applied in coded caching. 

\paragraph{Subsequent work.} In our work, we introduce the notion of task-based schemes for EIC, and develop heuristics for these schemes.  However, we left it as an open problem to understand the limitations of task-based schemes relative to other schemes.  Since our work first appeared, Haviv has solved this problem by giving tight bounds on the gap between task-based schemes and centralized schemes for EIC~\cite{haviv2019taskbased}.  Briefly, this work shows that there for any graph $G$, the length of the best task-based scheme is at most quadradically worse than the best scheme without the task-based restriction, and also shows that there exist graphs where this gap is asymptotically tight.

\section{Framework}\label{sec:framework}
In this section we formally describe our model for Embedded Index Coding.  

We assume that there is a set of $m$ data blocks, $\data \in (\bF_2^\ell)^m$, where each data block is an element of $\bF_2^\ell$; when convenient, we will view $\data \in \bF_2^{m \times \ell}$ as an $m\times \ell$ boolean matrix with the $m$ data blocks as rows.  These $m$ data blocks are stored on $n$ storage nodes; each node $i$ stores a subset of the data blocks, and some data blocks may be stored on multiple nodes.  We assume that each node can perform local computations and can broadcast information over an error-free channel to all the other nodes.  
In this work, we focus on a \em linear \em model, where each node is restricted to computing $\bF_2$-linear combinations of data blocks. 


An Embedded Index Coding (EIC) problem is defined in terms of which data blocks each node \emph{has} and \emph{needs}.  
It will be convenient to represent these ``has'' and ``needs'' relationships in terms of binary matrices $B$ and $R$ respectively.
\begin{definition}
An \emph{Embedded Index Coding (EIC) problem} is specified by a pair of matrices $B,R \in \mathbb{F}_2^{n\times m}$ s.t. $\emph{\supp}(B) \cap \emph{\supp}(R) = \emptyset$.
\end{definition}
Informally, the interpretation should be that in an EIC problem $\probset$, a node $u$ \emph{needs} block $a$ if $R_{ua} = 1$ and $\emph{has}$ block $b$ if $B_{ub} = 1$. 

Each node $u$ will broadcast a set of $b_u \in \mathbb{N}$ linear combinations of the blocks it has, and the goal is for each node to be able to recover all of the blocks that it needs.  We formalize this in the following definition. 

\begin{definition}\label{def:generalsolve}
For an EIC problem $\probset$ a \emph{linear broadcast solution} that  solves $\probset$ is a collection of matrices $\sendmat^{(1)},...,\sendmat^{(n)}$ and integers $h_1,...,h_n$ with $\sendmat^{(u)} \in \mathbb{F}_2^{h_u\times m}$ for each $u\in[n]$ so that:
\begin{itemize}
\item For each $u \in [n]$ and each $a \in [m]$ so that $B_{ua} = 0$, the $a^{th}$ column of $\sendmat^{(u)}$ is zero.
\item For each $u\in [n]$ and each $a \in [m]$ so that $R_{ua} = 1$, there is some vector $\decvec^{(u,a)} \in \mathbb{F}_2^{\sum_{\ell} h_{\ell}+m}$ so that 
$$\indvec_a = \decvec^{(u,a)}\cdot \left[\begin{array}{c}  \sendmat^{(1)} \\ \hline  \sendmat^{(2)} \\ \hline  \vdots \\ \hline  \sendmat^{(n)} \\\hline  \mathrm{diag}(B_u) \end{array}\right],  $$
where $B_u$ is the row of $B$ indexed by $u$ and $\mathrm{diag}(B_u)$ is the matrix with $B_u$ on the diagonal.
 Above, $\indvec_j$ denotes the $j^{th}$ standard basis vector.
\item The \emph{length} of an EIC solution is $\Sigma_ah_a$, the number of symbols broadcast. We also refer to this as the \emph{communication cost} of the solution.
\end{itemize}
\end{definition}

To use a linear broadcast solution, each node $u$ computes and broadcasts 
$ \sendmat^{(u)} \cdot \data, $
where we view $\data \in \bF_2^{m \times \ell}$ as a matrix whose rows are the data blocks.  
This can be computed locally because the only non-zero columns of $\sendmat^{(u)}$ correspond to non-zero entries of row $B_u$, i.e. blocks node $u$ has.

In order to decode the blocks it wants, each node $u$ uses the fact that
$$\text{block $a$} = \indvec_a \cdot \data = \decvec^{(u,a)}\cdot \left[\begin{array}{c}  \sendmat^{(1)} \\ \hline  \sendmat^{(2)} \\ \hline  \vdots \\ \hline  \sendmat^{(n)} \\\hline  \mathrm{diag}(B_u) \end{array}\right] \cdot \data,  $$
and thus block $a$ is a linear combination (given by $\decvec^{(u,a)}$) of the broadcasts $\sendmat^{(1)}  \data, \ldots, \sendmat^{(n)} \data$ that node $u$ recieves and the data blocks that $u$ already has. 

\subsection{Problem Graph and Problem Matrix}
We next define some representations of embedded index coding problems (extending the work of \cite{bar2006index}) which will be useful in studying the length of solutions and the construction of algorithms. 

We begin by a defining a graph $G$ which captures an EIC problem.
The vertices of $G$ will correspond to \em requirement pairs \em of the EIC problem, defined as follows.
\begin{definition}\label{def:pairset}
Given an EIC problem $\probset$, the set of \em requirement pairs \em for $\probset$ is 
 $P = \{(u,a) \in [n]\times[m]: R_{ua} = 1\}$. 
\end{definition}


Now we can formally define the problem graph $G$ for an EIC problem $\probset$.
\begin{definition}\label{def:probgraph}
Given an EIC problem $\probset$, the \emph{problem graph} $G = (V,E)$ corresponding to $\probset$ is the graph $G$ with vertices $V =\{v_{(u,a)} : (u,a) \in P\}$ and (directed) edges  $E = \{(v_{(u,a)},v_{(w,b)}) : B_{ub} = 1 \text{ or } a=b\}$.
\end{definition}
That is, for $(u,a)$ and $(w,b)$ in $P$, there are two reasons that there could be an edge from the vertex $v_{(u,a)}$ to the vertex $v_{(w,b)}$: either the node $u$ has the block $b$ that the node $w$ wants, or else the two blocks $a$ and $b$ are the same block.  As we will see, these two types of edges play two different roles.

Figure~\ref{fig:sideinfograph} shows two examples of problem graphs. In Figure~\ref{fig:probgraphnormal}, all edges indicate where a node has a block that another is requesting, i.e. cases where $(v_{(u,a)},v_{(w,b)}) \in E(G)$ because $B_{ub} =1$. In Figure~\ref{fig:probgraphnot1}, dashed edges indicate pairs of vertices which represent two requests for the same block, i.e. cases where $(v_{(u,a)},v_{(w,b)}),(v_{(w,b)},v_{(u,a)}) \in E(G)$ because $a=b$.

\begin{figure}
	\centering  
	\subfloat[{\label{fig:probgraphnormal}}]{\includegraphics[width=.3\textwidth]{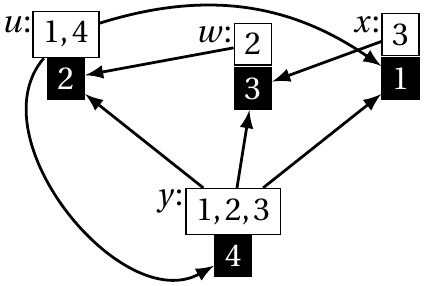}}\\
		\subfloat[{\label{fig:probgraphnot1}}]{\includegraphics[width=.38\textwidth]{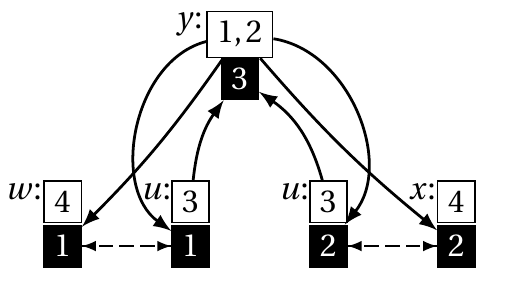}}
	\caption{Examples of a problem graph $G$ for nodes $u,w,x,y$ and data block $\mathcal{D}_1, \mathcal{D}_2, \mathcal{D}_3, \mathcal{D}_4$. Each pair of boxes is a vertex in $G$, where the black boxes contain indices of requested data blocks and the white boxes contain indices of side information blocks; each pair of boxes is labeled with the relevant node.  In (a), node $u$ is requesting block $\data_2$ and has blocks $\data_1,\data_4$ in its side information. In (b), node $u$ requests blocks $1$ and $2$, represented by two separate vertices. Since $w$ also requests block $1$ and $x$ also requests block $2$, we have a different type of edge (dashed) indicating vertices corresponding to the same requested block. }\label{fig:sideinfograph} 	
	
\end{figure}

\begin{definition}\label{def:fitmatrix}
Given a graph $G = (V,E)$, we say matrix $A\in \mathbb{F}_2^{|V|\times|V|}$ \emph{fits} $G$ if:
\begin{enumerate}
\item $A_{kk} = 1$ for all $k \in [|V|]$ and 
\item for any $k,\ell \in |V|$,$(k,\ell) \notin E$ implies that  $A_{k\ell}  = 0$.
\end{enumerate}
\end{definition}
Thus if $M$ is the adjacency matrix of $G$ and matrix $A$ fits $G$, the non-zero entries of $A$ (other than the diagonal) are a subset of the non-zero entries of $M$.

\begin{definition}\label{def:minrank}
The \emph{minrank} of a graph $G$ in field $\bF_2$, denoted $\minrk_{2}(G)$, is the rank of the lowest-rank matrix $A$ over $\bF_2$ which fits $G$:
$$\minrk_2(G) := \min \{\rk_2(A):A \text{ fits } G\}$$
\end{definition}

In Section~\ref{sec:minrkprobgraph}, we will show how our definition of a problem graph generalizes the \emph{side information graph} defined for index coding (that is, the centralized case of Figure~\ref{fig:modeltypes}(a), where each node requests a single unique block).  In this setting, it was shown in \cite{bar2006index} that $\minrk_2(G)$ is the length of the optimal index code.  We will show later how the minrank can also be used in computing solutions for EIC problems.

\subsection{Task-Based Solutions}\label{sec:tbsoln}

We now define a \em task-based solution, \em which is a particular type of solution to an embedded index coding problem.
As we will see, we can design efficient heuristics to find task-based solutions, and additionally task-based solutions may be more useful in settings with node failures.
\begin{definition}\label{def:newtask}
A task $\task$ is defined by a sender node $k$ and a set of pairs 
\[ M \subseteq \{ (u,a) \in P \,:\, B_{ka} = 1 \} \]
\end{definition}
Informally, if $T = (k,M)$ and $(u,a) \in M$, then this means that it is part of the  node $k$'s task to send the block $a$ to the node $u$.  Notice that this is not completely general: it rules out the possibility that the node $u$ could recover the block $a$ from two separate sender nodes.

A task-based solution is one built out of tasks.  We formally define this as follows.
\begin{definition}\label{def:tasksoln}
A \emph{task-based solution} to an EIC problem $\probset$ with requirement pairs $P$ is a linear broadcast solution  $\sendmat^{(1)},...,\sendmat^{(n)}$  so that $\sendmat^{(\ell)}\in \mathbb{F}_2^{h_{\ell}\times m}$, such that 
for each $(u,a) \in P$, 
there is an $\ell \in [n]$ and a coefficient vector $\decvec_{\ell}^{(u,a)}\in\mathbb{F}_2^{h_\ell+m}$ such that 
\[ \indvec_a = \decvec_{\ell}^{(u,a)}\cdot \left[\begin{array}{c} \sendmat^{(\ell)} \\ \hline  \mathrm{diag}(B_u) \end{array}\right] .\]
We say that such a node $\ell$ is \emph{responsible} for $(u,a)$ in the task $T$.
\end{definition}
Informally, a task-based solution is a linear solution in which each node $u$ decodes each requested block $a$ using only messages from one sender node $\ell$ who is responsible for $(u,a)$.  That is, $\ell$ broadcasts a vector $\sendmat^{(\ell)} \cdot \data$, and $u$ should be able to recover $a$ from this vector and its local side information.

A task-based solution to $\probset$ is related to the corresponding problem graph $G=(V,E)$ by specifying a partition of the vertices. Let  $N^+(v_{(u,a)}) \subseteq V$ denote the out-edge neighborhood of a vertex $v_{(u,a)} \in V$: that is,
\[ N^+(v_{(u,a)}) = \left\{ v_{(w,b)} \,:\, ( v_{(u,a)}, v_{(w,b)} ) \in E \right\}. \]

\begin{definition}\label{def:sendernborhood}
For an EIC problem $\probset$  with problem graph $G$,  define the \emph{sender neighborhood} of node $u\in [n]$ as:
 \[  N_u = \{v_{(w,b)}\in V: B_{ub} = 1\}. \]
\end{definition}
That is, the sender neighborhood $N_u$ of a  node $u$ is the set of vertices in $V$ corresponding to node-block pairs $(w,b)$ so that the node $u$ has the block $b$ (and thus could $u$ could send $b$ to $w$).
In terms of the problem graph $G$, $N_u \subseteq \cap_aN^+ (v_{(u,a)})$.

Figure~\ref{fig:probgraphns} shows examples of sender neighborhoods for the problem graph examples shown in  Figure~\ref{fig:probgraphnormal}. 

\begin{remark}\label{rem:icandidnc}
Finding tasks $(k,M)$ which maximize $|M|$ while minimizing total broadcast messages is a generalization of Instantly Decodable Network Codes (IDNCs)~\cite{le2013instantly}.  More precisely, solving the IDNC problem on sender neighborhood $N_k$ for some node $k$ finds the task $(k,M)$ with  maximal $|M|$ such that only one message needs to be broadcast by sender $k$ to satisfy all $(u,a) \in M$.
\end{remark}

\begin{remark}\label{rem:ldic}
Task-based solutions are also related to locally decodable index codes (LDICs)~\cite{natarajan2018locally}.  In an LDIC, a (centralized) index coding solution has locality $r$ if each node uses at most $r$ of the broadcast messages to decode any one block.  In the case that $r=1$, the natural generalization of LDICs to the decentralized setting is a special case of a task-based scheme.  When $r>1$, the two notions are different, but they have a similar flavor of restricting the information that can be used to reconstruct a single block.
\end{remark}

\begin{remark}\label{rem:icminiproblem}
Each node $k$ and its sender neighborhood $N_k$ (or any subset of $N_k$) together form an instance of an index coding problem with a single source: node $k$ is a source which has all blocks requested by nodes in $N_k$.
Thus the communication model is the same as in~\cite{bar2006index}, but it is not necessarily a single unicast problem (see Definition~\ref{def:onetoone}); that is, it is not the case that each node wants a unique block.
  \end{remark}


\begin{figure}
	\centering  
	\subfloat[{\label{fig:probgraphnnorm}}]{\includegraphics[width=.3\textwidth]{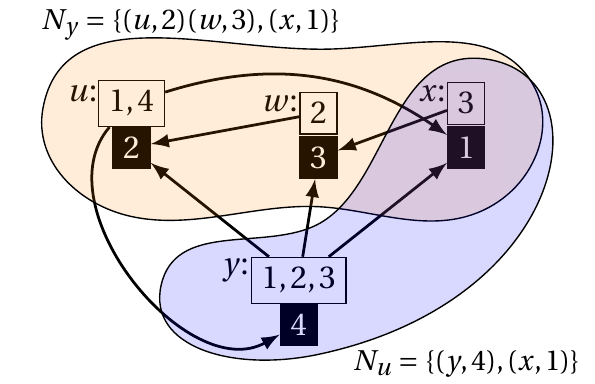}}\\
		\subfloat[{\label{fig:probgraphnon1t1n}}]{\includegraphics[width=.35\textwidth]{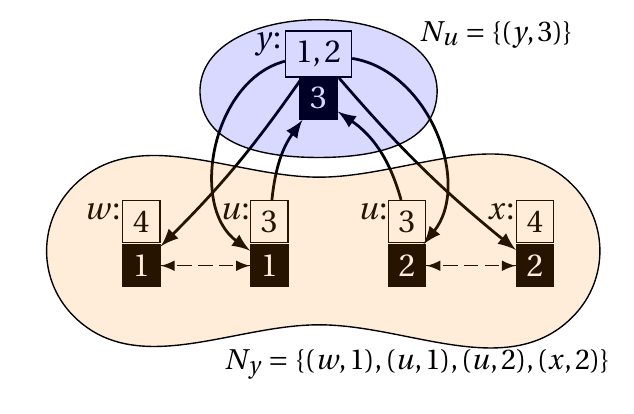}}
	\caption{ Examples from (a) Figure~\ref{fig:probgraphnormal} and (b) Figure~\ref{fig:probgraphnot1},  shown with sender neighborhoods of nodes $u$ and $y$. Note that in (b), the out-neighborhood of the vertex $(u,1)$ is $\{(w,1),(y,3)\}$ and the out-neighborhood of $(u,2)$ is $\{(x,2),(y,3)\}$ but the sender neighborhood (Definition~\ref{def:sendernborhood}) of $u$ is the intersection of these.
	}\label{fig:probgraphns} 	
\end{figure}

\begin{definition}\label{def:nborpartition}
Let $G$ be a problem graph with sender neighborhoods $N_1,...,N_n$. A \emph{neighborhood partition} is a set $\{\N_1,...,\N_n\}$ such that \begin{enumerate}
\item $\N_i \subseteq N_i$ for all $i = 1,...,n$,
\item $\N_i \cap \N_j = \emptyset$ for any $i,j \in [n]$,
\item and $\bigcup_{i \in[n]} \N_i = V(G)$.
\end{enumerate}
\end{definition}
Given an EIC problem with problem graph $G$ and task-based solution $T$, there is a corresponding neighborhood partition $\{\N_1,...,\N_n\}$: each vertex $v_{(u,a)}$ in $G$ belongs to the $\N_i$ such that $i \in [n]$ is responsible for $(u,a)$ in $T$. Furthermore, any neighborhood partition trivially corresponds to at least one task-based solution, in which each sender $i \in [n]$ broadcasts each block requested by a node in $\N_i$ as a separate message.


For example, there is a task-based solution for the EIC problem shown in Figure~\ref{fig:probgraphnormal}, using sender neighborhoods $\N_y = \{u,w,x\}$ and  $\N_u = \{y\}$. The messages for the task executed by node  $y$ are $\data_1\oplus \data_2$ and $\data_2\oplus \data_3$, and the message broadcast by node $u$ for its task is $\data_4$. Then nodes  $u$,  $w$, and  $x$ each decode their requested block from the task executed by node $y$, and node $y$ decodes its request from the task executed by node $u$.  Task-based solutions like this can be easier to compute than a distributed solution in general, and they allow some independence between nodes: in the example, nodes $w$ and $x$ do not need to wait for any node other than node $y$ to be able to decode their requested block.

\begin{remark}
While we only study task-based solutions on the EIC model, task-based solutions can also be defined  for multi-sender index coding in general.
\end{remark}

\subsection{Centralized Solutions}
We will later compare decentralized solutions to embedded index coding problems to an idealized centralized index coding solution. To that end, we define a solution to a an embedded index coding problem which assumes some oracle server exists with access to all of  $\data$ (and  has no requirements itself).

  \begin{definition}\label{def:centrsoln}
  For an EIC problem defined by $\probset$, a \emph{centralized linear broadcast solution} which solves $\probset$ is a matrix $\sendmat$ and $h\in \mathbb{N}$ with $\sendmat \in \mathbb{F}_2^{h\times m}$ such that for each $u \in [n]$ and each $a \in [m]$  with $R_{ua} = 1$, there  is some vector $\decvec^{(u,a)} \in \mathbb{F}_2^{h+m}$ so that $$\indvec_a = \decvec^{(u,a)} \cdot \left[\begin{array}{c}\sendmat \\\hline  \text{diag}(B_u) \end{array}\right].$$ 
  \end{definition}
Finally, we use the following symbols to denote the optimal lengths for each type of solution:
\begin{definition}
Let $\Cp$ denote the minimum length of a centralized linear solution to the EIC problem $\probset$ as defined in Definition~\ref{def:centrsoln}.
 
 Let $\Dp$ denote the minimum length of a decentralized linear broadcast solution to the EIC problem $\probset$ as  defined in Definition~\ref{def:generalsolve}. 
 
 Let $\Tp$ denote the minimum length of a decentralized and task-based solution to the EIC problem $\probset$ as defined in Definition~\ref{def:tasksoln}.
 \end{definition}

\section{Minimum Code Lengths  and Relationships}\label{sec:relationships}

In this section, we analyze the values of $\Cp$, $\Dp$, and $\Tp$ for a given $\probset$. We drop $\probset$ from the notation when comparing two of these under the same $\probset$ in general. While  it has been shown that graph-theoretic upper and lower bounds on minrank can have significant separation~\cite{shanmugam2014graph}, they are still useful in comparing the achievable minimum lengths in different solution types for EIC problems.

\subsection{$\Cp$ and minrank of the Problem Graph}\label{sec:minrkprobgraph}
First, we discuss an idealized centralized solution to an EIC problem, and introduce some useful machinery.

The work~\cite{bar2006index} defines the \em side information graph \em for an index coding problem.  We show how our problem graph is an effective generalization of the side information graph such that the same technique of using minrank to find an optimal centralized solution applies. 
The side information graph as defined by~\cite{bar2006index} is equivalent to a Problem Graph (Definition~\ref{def:probgraph})  for any \emph{single unicast} EIC problem $\probset$ (defined below).

\begin{definition}\label{def:onetoone}
An EIC specified  by $\probset$ is a \emph{single unicast} index coding problem if 
\begin{enumerate} \item every node requests exactly one data block and  
\item each data block is requested by exactly one node.
\end{enumerate}
\end{definition}
Figure~\ref{fig:probgraphnormal} shows the problem graph for a single unicast EIC; Figure~\ref{fig:probgraphnot1} shows the problem graph for an EIC which is not single unicast. 

We will generalize the following theorem, which restates Theorem 5 of \cite{bar2006index} using our definitions:
\begin{theorem}\label{thm:centrminrank}(Theorem 5  of \cite{bar2006index})
Given a single unicast EIC $\probset$ and the corresponding problem graph $G$, $\Cp = \minrk_2(G)$.
\end{theorem}

When a problem is not single unicast (in particular when the second condition of Definition~\ref{def:onetoone} does not hold) we constrain the minrank function  over  a subset of possible matrices, constructed as follows:

\begin{definition}
Given an EIC problem $\probset$ and a problem graph  $G= (V,E)$, we define the \emph{column repetition function} $\phi_{\probset}:\mathbb{F}_2^{|V|\times m} \to \mathbb{F}_2^{|V|\times |V|}$ as follows. Given a matrix $A\in \mathbb{F}_2^{|V|\times m}$, construct a matrix $\phi_{\probset}(A) \in \mathbb{F}_2^{|V|\times |V|}$ so that for $v_{(u,a)} \in V$,  the  column of $\phi(A)$ indexed by $v_{(u,a)}$ is equal to the $a^{th}$ column of $A$.
Additionally, we will denote the image of $\phi_{\probset}$ by $\Aprob= \{A': \exists A \text{ s.t. } \phi_{\probset}(A)=A'\}$.
\end{definition}

\begin{remark}\label{rem:preserverank}
The function $\phi_{\probset}$ preserves the rank of a matrix, since it just inserts duplicates of columns.   That is, $\rk_2(\phi_{\probset}(A) )= \rk_2(A)$.
\end{remark}

 For an EIC problem $\probset$, we will use the set $\Aprob$ to restrict the domain of minrank, resulting in the \em restricted-minrank:\em 
\begin{definition}
The \emph{restricted minrank} of a graph $G = (V,E)$ in the field $\mathbb{F}_2$ over set of matrices $\mathcal{A}\subseteq \mathbb{F}_2^{|V|\times |V|}$, denoted $\rminrk_2(G,\mathcal{A})$, is the rank of the lowest-rank matrix $A'\in\mathcal{A}$ which fits $G$:
$$\rminrk_2(G,\mathcal{A}) = \min\{\rk_2(A'): A' \in \mathcal{A} \land A' \text{ fits }G\}.$$
\end{definition}

\begin{lemma}\label{lem:basisdecode}
Let $G$ be the problem graph for an EIC problem defined by $\probset$. Let $A'$ be a matrix that fits $G$, and assume that $A' = \phi_{\probset}(A)$ for some matrix $A \in \mathbb{F}_2^{|V|\times m}.$ 
Suppose that $A^{(r)} \in \mathbb{F}_2^{r \times m}$ is a matrix whose rows are $r$ rows of $A$ which span the rowspace of $A$; thus, the rowspace of $A^{(r)}$ is equal to that of $A$.
Then $A^{(r)}$ 
is a centralized linear broadcast solution  to $\probset$.
 \end{lemma}
\begin{proof}
Let $r:=\rk_2(A) = \rk_2(A')$. Without loss of generality, suppose that the first $r$ rows $A_1,...,A_r$ span the rowspace $A$, so any node $u$ and one of its corresponding block request rows $A_{(u,a)}$ can compute 
$A_{(u,a)} = \sum_{i=1}^r \lambda_i^{(u,a)}A_{i}$.  For ease of notation let $\ell :=(u,a)$, so $A_\ell = A_{(u,a)}$ denotes the row of $A$ indexed by node  $u$ requesting block $a$ (Definition~\ref{def:fitmatrix}).  

Let $\beta$ be the matrix 
$$\beta  = \begin{bmatrix} --A_1--\\  \vdots \\ --A_r--\end{bmatrix},$$ so the rows of $\beta  \data$ are  the  encoded messages $\{A_1\cdot\data,A_2\cdot \data,...,A_r\cdot \data\}$.
Then $[\lambda^{(\ell)}_1,...,\lambda^{(\ell)}_r] \cdot  (\beta  \data) =A_{\ell}\cdot \data$, so from the encoded messages node $u$ can compute $ A_{\ell}\cdot \data$.

We next define the vector $\mu \in \mathbb{F}_2^m$: let $\mu_b = A_{\ell b}$ if $B_{ub} = 1$, otherwise let $\mu_b= 0$ (equivalently, $\mu = B_u\odot A_{\ell}$\footnote{Here, we use  $\odot$ to  denote \em Hadamard \em or entry-wise product}).  By definition of $B$, node $u$ can compute $\indvec_b\data$ for any $b\in [m]$ such that $B_{ub} =  1$ and thus $u$ can compute $\mu\cdot \data = (A_{\ell}-\indvec_a)\cdot  \data$.

Then we construct the decoding vector $\decvec^{(u,a)}$: 
$$\decvec^{(u,a)} = [\lambda_1^{(\ell)}\dots \lambda_r^{(\ell)}\  \mu_1\dots \mu_m]$$

so that decoding is done by computing 
\[  \decvec^{(u,a)}\cdot  \left[\begin{array}{c} \beta \\\hline  \text{diag}(B_u)\end{array}\right]\cdot \data  = [\lambda^{(\ell)}_1,...,\lambda^{(\ell)}_r] \cdot  (\beta  \data)  - \mu\cdot \data=A_{\ell}\cdot \data - (A_{\ell}-\indvec_a)\data = \indvec_a \data. \]

\end{proof}
We next generalize Theorem~\ref{thm:centrminrank}  to EIC problems which are not single unicast:
\begin{theorem}\label{thm:rminrankopt}
Given an EIC problem $\probset$, corresponding problem graph $G$, and column repetition function $\phi_{\probset}$ with range $\Aprob$,
$$\Cp = \rminrk_2(G,\Aprob).$$
\end{theorem}
\begin{proof}
Let $A'\in \Aprob$ be the matrix of lowest rank in $\Aprob$ that fits $G$ and let $r:=\rk_2(A') = \rminrk_2(G,\Aprob)$. Let $A \in \mathbb{F}_2^{|V|\times m}$ such that $\phi_{\probset}(A) = A'$. By Lemma~\ref{lem:basisdecode}, a matrix $A^{(r)}$ composed of $r$ linearly independent rows of $A$ is a centralized linear solution to $\probset$ of length $r$ (by the choice of $r$, the rowspan of $A^{(r)}$ equals the rowspan of $A$). Since a centralized source is able to construct each of these messages for this solution (that is, the rows of matrix $A^{(r)}\data$) we conclude that 
\[ \Cp\leq r = \rminrk_2(G,\Aprob). \]

For the other direction,
suppose that $Z \in \mathbb{F}_2^{s \times m}$ is a linear solution to $\probset$ for some $s \in \mathbb{Z}^+$.
Let $\mathbf{z}_i  \in \mathbb{F}_2^m$ denote the $i^{th}$ row of $Z$. 
We will show the row span of $Z$ contains the row span of some matrix $A$ such that $A':=\phi_{\probset}(A)$ fits  $G$. Consider some $(u,a)\in P$. By the definition of a linear solution, there exists some vector $\decvec^{(u,a)}$ such that 
$$\indvec_a = \decvec^{(u,a)}\cdot  \left[\begin{array}{c} Z \\\hline  \text{diag}(B_u)\end{array}\right].$$
Write 
$$\decvec^{(u,a)} = [\lambda_1^{(u,a)}\dots \lambda_s^{(u,a)}\  \mu_1\dots \mu_m]$$  
for some $\lambda_i^{(u,a)}, \mu_j \in \mathbb{F}_2$, so that 
$$\indvec_a = \sum_{i=1}^s \lambda_i^{(u,a)}\mathbf{z}_i+\sum_{j=1}^m \mu_jB_{uj}\indvec_j.$$
%
 Let $A_{(u,a)} \in \mathbb{F}_2^m$ be the vector
\[ A_{(u,a)} =  \indvec_a-\sum_{j=1}^m \mu_jB_{uj}\indvec_j = \sum_{j=1}^s \lambda_j^{(u,a)} \mathbf{z}_j.\] 
Then 
 $A_{(u,a)}$ is in the row span of $Z$, and moreover the $a$'th entry of $A_{(u,a)}$ satisfies
\[ A_{(u,a),a} = 1. \]
Additionally, for any block $b\in [m]$ with $b \neq a$ such that $B_{ub} =  0$, 
\[ A_{(u,a), b} = (\indvec_a)_b - \mu_b B_{ub} = 0. \]
Let $A \in \mathbb{F}_2^{|P| \times m}$ be the matrix whose rows are given by $A_{(u,a)}$ for $(u,a) \in P$.
Let 
 $A':=\phi_{\probset}(A)$.
We claim that $A'$ fits $G$. 
Indeed, we have for all $(u,a) \in P$ that
\[ A'_{(u,a),(u,a)} = A_{(u,a),a} = 1 \]
by the above, and so the first requirement of Definition~\ref{def:fitmatrix} is met.

To see the second requirement of Definition~\ref{def:fitmatrix}, first note that for all $b \neq a$, we have
\[ A'_{(u,a),(w,b)} = A_{(u,a),b} \] 
which by the above is non-zero only if $B_{ub} = 1$; that is only if there is an edge (of the ``first type'') from $v_{(u,a)}$ to $v_{(w,b)}$ in $G$.  
Second, there is always an edge (of the ``second type'') from $v_{(u,a)}$ to $v_{(w,a)}$.
Thus, the only non-zero off-diagonal entries of $A'$ correspond to edges in $G$, and the second requirement of Definition~\ref{def:fitmatrix} is satisfied.


Thus, for any linear solution $Z$ of length $s$, there is a matrix $A$ so that row span of $Z$ contains the row span of $A$ and so that $A' = \phi_{\probset}(A)$ fits $G$.  Thus,
\[ \Cp \geq s \geq \rk_2(A) = \rk_2(A') \geq \rminrk_2(G,\Aprob). \]
This completes the proof.

\end{proof}

Because the  minrank gives the optimal linear solution for a centralized sender with all data blocks, our definition of the problem graph is a natural extension of index coding and the side information graph to the embedded index coding model.  In the following it will be helpful to use the following theorem from~\cite{bar2006index} relating minrank to some other standard graph properties.  For a graph $G$, the \em chromatic number \em $\chi(G)$ is the minimum number of colors required to color the vertices of $G$ so that no neighboring vertices have the same color.  The \em clique number \em $\omega(G)$ is the size of the largest clique in $G$. The independence number, denoted $\alpha(G)$, is the set of the largest independent set in $G$, so  $\alpha(G) = \omega(\overline{G})$. 

\begin{theorem}[\cite{bar2006index}]\label{thm:minrkbounds}
$\omega(\overline{G})\leq \minrk_2(G)\leq  \chi(\overline{G})$.
\end{theorem}
These bounds also apply to our  restricted version of minrank:
\begin{cor}
Given an EIC problem $\probset$, a corresponding problem graph $G = (V,E)$, and the column repetition function $\phi_{\probset}$ with range $\Aprob$, we have:
$$\omega(\overline{G})\leq  \minrk_2(G) \leq \rminrk_2(G,\Aprob)\leq \chi(\overline{G}).$$
\end{cor}
\begin{proof}
Since $\rminrk_2$ is a minimization over a smaller set of matrices than $\minrk_2$, clearly $\minrk_2(G) \leq \rminrk_2(G,\mathcal{A})$. 
Thus 
\[ \omega( \overline{G} ) \leq \minrk_2(G) \leq \rminrk_2(G,\mathcal{A}) \]
follows from Theorem~\ref{thm:minrkbounds}.
Using a similar approach as in~\cite{bar2006index}, we show the final inequality by describing a matrix $A' \in \mathcal{A}$ such that $\rk_2(A') \leq \chi(\overline{G})$.  

By the definition of chromatic number, there is a partition of $V$ into sets $C_1, \ldots, C_{\chi(\overline{G})}$ so that each $C_i$ forms a clique in $G$.
Let $C\subseteq V$ be a clique from such a partition. 
Define a vector $\mathbf{c}^{(C)} \in \mathbb{F}_2^m$ so that the $a$'th entry of $\mathbf{c}^{(C)}$ is given by
\[ c_a^{(C)} = \begin{cases} 1 & \exists u \in [n] \text{ so that } v_{(u,a)} \in C \\ 0 & \text{else} \end{cases} \]


Now, define a matrix $A \in \mathbb{F}_2^{|P| \times m}$ with rows indexed by elements of $P$, so that 
if $v_{(u,a)}$ is in the clique $C$, then
\[ A_{(u,a)} = \mathbf{c}^{(C)}. \]
 Let $A' = \phi_{\probset}(A)$.  Thus
\[ \rk_2(A') = \rk_2(A) \leq \chi(\overline{G}) \]
since there are only $\chi(\overline{G})$ distinct rows of $A$.

Now we just need to show that $A'$ fits $G$. 
 Consider some row $A'_{(u,a)}$, where $A_{(u,a)} = {\bf c}^{(C)}$ for some clique $C$, and choose some $(w,b)\neq (u,a)$ such that $A_{(u,a)(w,b)}=1$. If $a=b$, then $(v_{(u,a)},v_{(w,b)})\in E$ is an edge of the ``second type.''
On the other hand, if $a \neq b$, then by the definition of $\phi_{\probset}$, $c_b^{(C)} = 1$, so there exists some $x \in [n]$ so that $v_{(x,b)} \in C$.  Since $C$ is a clique, 
$(v_{(u,a)},v_{(x,b)})\in E$. By the definition of problem graph this is an edge of the ``first type,'' so $B_{ub}=1$, so we also have $(v_{(u,a)},v_{(w,b)})\in E$.  Thus any non-zero off-diagonal entry of $A$ corresponds to an edge in $G$.
Moreovoer, the  diagonal entries of $A'$ are 
\[ A'_{(u,a),(u,a)} = A_{(u,a),a)} = c_a^{(C)} \]
where $v_{(u,a)} \in C$, so this is $1$ from the definition of $c_a^{(C)}$.
Thus, $A'$ fits $G$.
\end{proof}

\subsection{Cost of Decentralization: $\D \leq 2\C$}

It can easily be  seen that $\C \leq \D$, that is, that the minimum length of a decentralized embedded index code is at least the minimum length of a centralized solution.
Indeed, the $\D$ messages transmitted in the decentralized solution can all be constructed by a centralized source which has access to all data blocks.  Thus we are interested in how much larger $\D$ can be than $\C$.
In fact, we show that it is no more than a factor of $2$ worse:
\begin{theorem}\label{thm:2xdecentr}
Given an EIC problem $\probset$, $\Dp \leq 2\cdot \Cp$.
\end{theorem}
\begin{proof}

Let $P$ be the set of requirement pairs for $\probset$.
Let $G = (V,E)$ be the problem graph for $\probset$ and let $\phi_{\probset}$ be the corresponding column expansion function, with image $\Aprob$. By Theorem~\ref{thm:rminrankopt},  $\Cp = \rminrk_2(G,\Aprob)$. 
Let $A' \in \Aprob$ be a matrix with $A' =\phi_{\probset}(A)$ for some $A \in \mathbb{F}_2^{|V| \times m}$ so that
\[ \rk_2(A') = \rk_2(a) = \Cp =: r \]
and so that $A'$ fits $G$.
By Lemma~\ref{lem:basisdecode}, there is a matrix $A^{(r)}$ with rows $A_1, \ldots, A_r$ that is a centralized linear broadcast solution to $\probset$. We will show how to simulate this centralized solution using only $2r$ messages.

Since $A_1, \ldots, A_r$ are rows of $A$, they correspond to requirement pairs in $P$.  Fix $\ell \in [r]$ and suppose that $A_\ell$ corresponds to $(u,a) \in P$.   Since $A'$ fits $G$, the diagonal entries of $A'$ are non-zero.  This means that $A_{\ell, a} \neq 0$.  Further, for $b \neq a$, if $A_{\ell, b} \neq 0$ then there is an edge of the ``first type'' in $G$: that is, $B_{ub} = 1$, which means that node $u$ has block $b$.
%
%
Thus, node $u$ is able to compute
$$\sum_{b:B_{ub} = 1}A_{\ell b} \indvec_b \cdot \data = A_\ell\cdot \data - \indvec_a \cdot \data.$$ 

The decentralized scheme is then as follows.  For each $\ell \in [r]$ corresponding to $(u,a)$, we have two broadcasts:
\begin{enumerate}
\item Node $u$ broadcasts $\sum_{b:B_{ub} = 1}A_{\ell b} \indvec_b \cdot \data$.  That is, $A_{\ell}$ is a row of $\beta^{(u)}$.
\item Fix any other node $w$ so that $B_{wa} = 1$.  Then node $w$ broadcasts $\indvec_a \cdot \data$.  That is, $\indvec_a$ is a row of $\beta^{(w)}$.
\end{enumerate}

Now every node can add together the two broadcasts corresponding to $\ell \in [r]$ to obtain $A_\ell \cdot \data$.  
Since $A^{(r)}$ is a linear centralized solution to $\probset$, this scheme is a linear centralized solution to $\probset$.
\end{proof}
We note that the proof of Theorem~\ref{thm:2xdecentr} crucially uses the EIC formulation; this shows why considering EIC separately as a special case of multi-sender index coding can be valuable.

\subsection{Cost of Task-Based Solutions: Upper Bound for $\T$}

 We first show how the minrank can be used to re-formulate the length $\T$ of the optimal task-based solution. 
Let $(R,B)$ be an EIC problem, with problem graph $G = (V,E)$.
Recall from Definition~\ref{def:nborpartition} that, given a task-based solution $T$, the neighborhood partition $\{\N_1, \ldots, \N_n\}$ is a partition of $V$ so that $\N_w \subseteq N_w$ is the set of vertices $v_{(u,a)}$ so that $w$ is responsible for $(u,a)$ in $T$.  

For $\N_w \subseteq N_w$ corresponding to a task-based solution $T$, let $G|_{\N_w}$ denote the induced subgraph of $G$ on the vertices $\N_w$. 
As per Remark~\ref{rem:icminiproblem}, each $\N_w$ corresponds to an EIC problem $(R^{(w)}, B^{(w)})$, over the set of blocks $\{a \in [m]:\exists u\in[n]\text{ s.t. } v_{(u,a)} \in \N_w\}$. Thus by definition of $\N_w$, node $w$ has all blocks used in problem $(R^{(w)}, B^{(w)})$ and any centralized solution to $(R^{(w)}, B^{(w)})$ can be broadcast by $w$. Note that such a solution can easily be used as a self-contained part of a solution to the problem $(R,B)$ with the full set of $m$ blocks. To do so, we just insert zeros in encoding and decoding vectors for blocks in $[m]$ not used in $(R^{(w)}, B^{(w)})$. Then the messages of the solution to $(R^{(w)}, B^{(w)})$ can be used by vertices of $\N_w$ as in the subproblem. 

 We first show how solutions to these subproblems can be used as building blocks for task-based solutions.
\begin{lemma}\label{lem:ctotb}
Let $G$ be a problem graph for EIC problem $(R,B)$. Let $\{\N_1,....,\N_n\}$ be a neighborhood partition. Let $G|_{\N_w}$ be the subgraph of $G$ induced by $\N_w$  and let $(R^{(w)},B^{(w)})$ be the problem with problem graph $G|_{\N_w}$ for all $w \in [n]$. Then any set of solutions $\{\beta^{(w)} \in \mathbb{F}^{h_w\times m}: w \in [n], h_w \in \mathbb{Z}^+\}$ for problems $\{(R^{(w)},B^{(w)}):w \in[n]\}$ forms a task-based solution to $(R,B)$  with length $\sum_{i=1}^n h_i$.
\end{lemma}
\begin{proof}
For each vertex $v_{(u,a)}$ of $G$ there is some $\N_w$ such that $v_{(u,a)} \in \N_w$. Let $\beta^{(w)}$ be the centralized linear broadcast solution to EIC problem $(R^{(w)},B^{(w)})$, where $\beta^{(w)} \in \mathbb{F}^{h_w\times m}$ for some $h_w \in \mathbb{Z}^+$. Then there exists some $\decvec^{(u,a)}$ such that $\indvec_a = \decvec^{(u,a)} \cdot \left[\begin{array}{c}\sendmat^{(w)} \\\hline  \text{diag}(B_u) \end{array}\right]$. Since there is such a  vertex $v_{(u,a)}$ for each $(u,a) \in P$, all requests in $P$ are satisfied in this way by some $\beta^{(w)}\cdot \data$. By definition of $(R^{(w)},B^{(w)})$, each $\beta^{(w)}\cdot \data$ for $w \in[n]$ can be broadcast by node $w$.
 Thus $\beta^{(1)},...,\beta^{(n)}$ forms a task-based solution to $(R,B)$ with length $\sum_{i=1}^n h_i$.

\end{proof}
We can then compute the length of an optimal task-based solution, $\T$, in terms of neighborhood partitions.
\begin{lemma}\label{lem:neighborminrank}
Given an EIC problem  $\probset$, let $\mathscr{N}$ be the set of all possible neighborhood partitions (as in Definition~\ref{def:nborpartition}). 
For $\{\N_1,...,\N_n\} \in\mathscr{N}$, let $(R^{(w)},B^{(w)})$ be the EIC problem induced by $\N_w$.
Then
$$\displaystyle\Tp = \min_{\{\N_1,...,\N_n\} \in\mathscr{N}} \sum_{i=1}^n \rminrk_2(G|_{\N_i},\mathcal{A}_{(R^{(i)},B^{(i)})}).$$
 \end{lemma}
 \begin{proof}
We first show that $\Tp \leq  \min_{\{\N_1,...,\N_n\} \in\mathscr{N}} \sum_{i=1}^n \rminrk_2(G|_{\N_i},\mathcal{A}_{(R^{(i)},B^{(i)})})$. Consider the neighborhood partition $\{\N_1,...,\N_n\}$ which minimizes $ \sum_{i=1}^n \rminrk_2(G|_{\N_i},\mathcal{A}_{(R^{(i)},B^{(i)})})$. A possible task based solution $T$ can be constructed by optimally solving the centralized index coding problem $(R^{(w)},B^{(w)})$ defined by each $G|_{\N_w}$ with sending node $w$, as shown in Lemma~\ref{lem:ctotb}. By Theorem~\ref{thm:rminrankopt}, each centralized subproblem solution has length $\rminrk_2(G|_{\N_w}, \mathcal{A}_{(R^{(w)},B^{(w)})})$, so the total length of $T$ is $\sum_{i=1}^n \rminrk_2(G|_{\N_i},\mathcal{A}_{(R^{(i)},B^{(i)})})$.

Next we show $\min_{\{\N_1,...,\N_n\} \in\mathscr{N}} \sum_{i=1}^n \rminrk_2(G|_{\N_i},\mathcal{A}_{(R^{(i)},B^{(i)})}) \leq \Tp$. Let $T$ be the optimal task-based solution, with length $\Tp$. Construct $\{\N_1,...,\N_n\}$ so that
\[ \N_w = \left\{ v_{(u,a)} \,:\, \text{$w$ is responsible for $(u,a)$ in $T$} \right\}. \]

By the definition of a task-based solution, each vertex is assigned to exactly one such $\N_w$, so we have a neighborhood partition $\{\N_1,...,\N_n\}\in \mathscr{N}$.  The centralized index coding problems $(R^{(w)},B^{(w)})$ for each  $w \in [n]$ have problem graphs $G|_{\N_w}$ and optimal solutions of length $\rminrk_2(G|_{\N_w},\mathcal{A}_{(R^{(w)},B^{(w)})})$ (Theorem~\ref{thm:rminrankopt}). If $\Tp < \sum_{i=1}^n \rminrk_2(G|_{\N_i},\mathcal{A}_{(R^{(i)},B^{(i)})})$, then the solution to $(R^{(w)},B^{(w)})$ for some $w \in [n]$ must have length strictly less than $\rminrk_2(G|_{\N_w},\mathcal{A}_{(R^{(w)},B^{(w)})})$. This contradicts Theorem~\ref{thm:rminrankopt}. Thus $\Tp \geq \min_{\{\N_1,...,\N_n\} \in\mathscr{N}} \sum_{i=1}^n \rminrk_2(G|_{\N_i},\mathcal{A}_{(R^{(i)},B^{(i)})}).$

  \end{proof}

The following upper bound on $\Tp$ follows from 
Lemma~\ref{lem:neighborminrank} and Theorem~\ref{thm:minrkbounds}.
  \begin{lemma}\label{lem:nborupper}
  Given an EIC problem defined by $\probset$, let $\mathscr{N}$ be the set of all possible neighborhood partitions (Definition~\ref{def:nborpartition}).  Then
\begin{equation}\label{eq:nborupper}
\displaystyle\Tp  \leq  \min_{\{\N_1,...,\N_n\}\in\mathscr{N}} \sum_{i=1}^n \chi(\overline{G}|_{\N_i}).
\end{equation}
 \end{lemma}
 
 In Section~\ref{sec:tapprox}, we will use Lemma~\ref{lem:nborupper} to develop algorithms to approximate the optimal neighborhood partition (in the sense that the right hand side of \eqref{eq:nborupper} is minimized), by reducing the problem to the minimum cover problem.
 
 \subsection{An example where $\C \neq \D \neq  \T$}

\begin{figure}
	\centering  
	\subfloat[{\label{fig:graphexample}}]{\includegraphics[width=.26\textwidth]{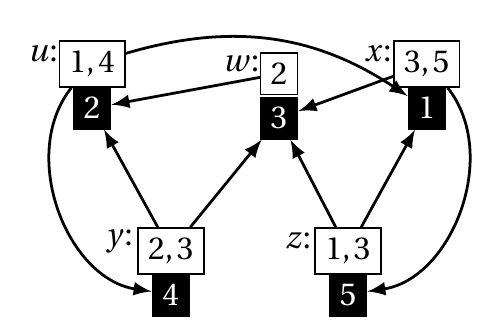}}
	\hspace{1cm}
		\subfloat[{\label{fig:gexampleneighbor}}]{\includegraphics[width=.45\textwidth]{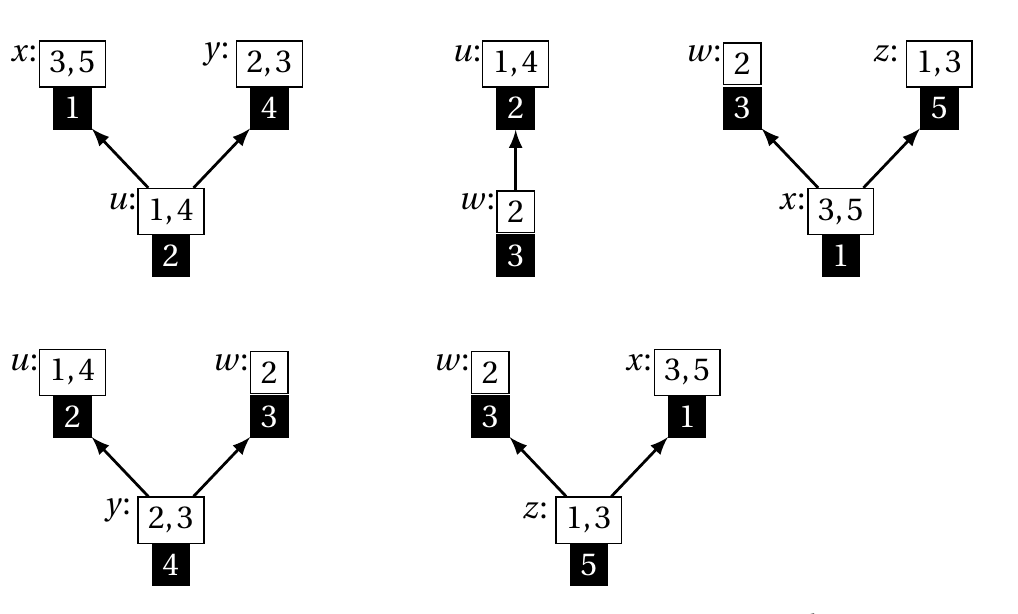}}
	\caption{ (a) Example of a problem graph for a single unicast EIC problem that has $\C < \D < \T$. Each pair of boxes is a vertex in the problem graph, where black boxes contain indices of requested data blocks and the white boxes contain indices of side information blocks; the label $u,w,x,y,z$ indicates which node the vertex corresponds to. (b) The subgraphs induced by each node along with its sender neighborhood.}  
\end{figure}

Figure~\ref{fig:graphexample} is an example of an EIC problem $\probset$ for which $\C < \D < \T$.  First consider $\C$: by inspection, a central source  with  all data blocks could send messages $\data_2\oplus \data_4$, $\data_3$,  and $\data_1\oplus \data_5$ so that all five nodes can decode their requested block, but no combination of fewer messages suffices.  Thus, $\C = 3$. 

Next consider $\D$. A solution of minimum-length is: node $y$  broadcasts $\data_2\oplus \data_3$, node $z$ broadcasts $\data_3 \oplus \data_1$, node $u$ broadcasts $\data_4$, and node $x$ broadcasts $\data_5$. It can be checked that this is indeed a minimum-length solution. Thus, $\D = 4$. 

Finally, consider $\T$. Then out-neighborhoods of each node, as shown in Figure~\ref{fig:gexampleneighbor}, are the subgraphs over which we can apply index coding (Remark~\ref{rem:icminiproblem}). In particular, we construct  the neighborhood partition from these subgraphs (Definition~\ref{def:nborpartition}). Since the graph induced by each neighborhood is acyclic, as shown in~\cite{tahmasbi2015critical} there is no way to do any non-trivial coding in any subgraph to a code shorter than the uncoded solution. Thus any task based solution requires that all blocks be broadcast uncoded.  Since there are five blocks that need to be sent, we have $\T = 5$.

\subsection{Separations between $\T$ and $\C$}

 We next give a condition on the problem graph $G$ which guarantees that $\T$ is strictly larger than $\C$, with a gap as big as the gaps from the graph-theoretic minrank bounds. 
 \begin{lemma}\label{lem:chigap}
Given an EIC problem defined by $\probset$ and corresponding problem graph  $G$, If $(\chi(G)-1)\chi(\overline{G})<|V|$ then there is an  optimal task-based solution with neighborhood partition $\{\N_1,...,\N_n\}$ so that
$$\Cp \leq \chi(\overline{G}) < \sum_{i=1}^n \omega(\overline{G}|_{\N_i}) \leq \T_{\probset}$$
\end{lemma}
\begin{proof}
Let $V:=V(G)$.
First note that for any graph $G$, $\alpha(G) = \omega(\overline{G})$  and $\frac{|V|}{\chi(G)}  \leq \alpha(G)$ (see, e.g.,~\cite{west1996introduction}). Consider coloring each graph $G|_{\N_i}$ (induced by an element of the neighborhood partition) individually, compared to coloring  all of $G$ at once. Since every node in $\N_i$ shares a neighbor in $G$ (i.e. any of  the vertices $v_{(i,a)} \in V$ for some $a\in[m]$) there is a color in the minimum coloring of $G$ not necessary to color $G|_{\N_i}$. Thus $\chi(\N_i)\leq  \chi(G)-1$. Putting these steps together:

 $$ \sum_{i=1}^n \omega(\overline{G}|_{\N_i}) \geq  \sum_{i=1}^n \frac{|V(\Gi)|}{\chi(\Gi)} \geq \sum_{i=1}^n \frac{|V(\Gi)|}{\chi(G)-1} = \frac{|V|}{\chi(G)-1} $$
Since by assumption we have
  $\frac{|V|}{\chi(G)-1}> \chi(\overline{G})$, this proves the claim. Note that $ \sum_{i=1}^n \omega(\overline{G}|_{\N_i}) \geq  \sum_{i=1}^n \frac{|V(\Gi)|}{\chi(\Gi)} $ follows from applying $\omega(\overline{G}) \geq \frac{|V|}{\chi(G)}$ to each induced graph $\Gi$. Additionally, $\sum_{i=1}^n |V(\Gi)| = |V|$ leads to the equality on the right because the graphs in the set $\{\Gi:i \in [n]\}$ are induced by corresponding elements of the vertex partition $\{\N_1,...,\N_n\}$.
\end{proof}

Lemma~\ref{lem:chigap} establishes a gap between $\Cp$ and $\T_{\probset}$ whenever 
$|V| > (\chi(G)-1)(\chi(\overline{G}))$, so we note here a few small graphs which illustrate how this may or may not be the case.
First, we note that for all
graphs it is true that $|V| \leq \chi(G)\chi(\overline{G})$. 
For cliques and graphs consisting of multiple disconnected cliques, $|V| =  \chi(G)\chi(\overline{G}) > (\chi(G)-1)(\chi(\overline{G}))$ so Lemma~\ref{lem:chigap} establishes a gap. On the other hand, for (directed) cycles, Lemma~\ref{lem:chigap} does not establish a gap: $\chi(G) = 2$ or $3$ and $\chi(\overline{G})  = n$ so $(\chi(G)-1)(\chi(\overline{G}))\not < n$.

\section{Algorithms}\label{sec:algos}

In this section, we use results from the previous section to design two heuristics for finding good EIC solutions.  We also demonstrate empirically that our algorithms perform well.  

First, we use Theorem~\ref{thm:2xdecentr} to give an algorithm which produces an EIC solution that is optimal within a factor of two.  We show empirically that our algorithm is faster (more precisely, has a smaller search space) than the algorithm of \cite{li2018multi}.  We note that our algorithm is tailored for EIC while the approach of \cite{li2018multi} works more generally in the multi-sender model.  This demonstrates the value of focusing on EIC as a special case.

Second, we use Lemma~\ref{lem:nborupper} to give a heuristic algorithm to design a task-based scheme for an EIC problem.  We show empirically that the quality of solution returned by our algorithm is within a small constant factor (at most $1.4$ in our experiments) of the optimal \em centralized \em scheme.

We describe both of these in more detail below.

\subsection{Approximating $\D$}\label{sec:approxd}
The proof of Theorem~\ref{thm:2xdecentr} gives an algorithm to approximate the optimal decentralized solution to an EIC problem,
which we detail in Algorithm~\ref{alg:2xdecentr}.
Algorithm~\ref{alg:2xdecentr} first computes the exact optimal centralized solution with length $\Cp$ and then uses the transformation outlined in the proof of Theorem~\ref{thm:2xdecentr} to arrive at a decentralized solution with length at most $2 \Cp$.  We note that in practice the optimal centralized solution could also be approximated, leading to a decentralized solution of length at most twice the cost of the approximation. 

\begin{algo}\label{alg:2xdecentr}
Given an EIC problem $\probset$:
\begin{enumerate}
\item Construct the problem graph $G$
\item Find $A \in \mathbb{F}_2^{|V|\times m}$ such that $\rk_2(A)= \rk_2(\phi_{\probset}(A)) = \rminrk_2(G,\Aprob)$\label{step:minrk}
\label{alg2x:minrkstep} 
\begin{enumerate}
\item Let $r:=rk_2(A)$
\item Let $A_1,...,A_r$ be linearly independent rows of $A$
\end{enumerate}

\item For each $A_{\ell} \in\{A_1...A_r\}$:
\begin{enumerate}
\item Let $\ell = (u,a)$
\item Node $u$: compute and broadcast $\sum_{\substack{k = (w,a); k\in[|V|]: B_{ui} = 1}} A_{\ell k} \data_a $ 
\item For each node $w$ s.t. $B_{(w,a)} = 1$: compute and broadcast $(\sum_{k = (w,a);w \in [n]} A_{\ell k})\data_a$
\end{enumerate}
\end{enumerate}
\end{algo}

We compare Algorithm~\ref{alg:2xdecentr}
 to the algorithms in previous work~\cite{li2018multi,kim2019linear}, which apply more generally to any multi-sender index coding problem but can in particular be applied to EIC. Since the algorithms of~\cite{li2018multi,kim2019linear} only apply to EIC problems in which each node is requesting a single block, we restrict our analysis to that case. In this case, applying $\phi_{\probset}$ has no effect (that is, $A = \phi_{\probset}(A)$), so the computation of $\rminrk_2(G,\Aprob)$ in step~\ref{step:minrk} of Algorithm~\ref{alg:2xdecentr} is equivalent to computing $\minrk_2(G)$.

To compare the complexity of these algorithms, we observe that the main computational task of both methods is computing the minrank of a graph, by searching over a set of possible fitting matrices. 
In practice, we may wish to use a heuristic to approximate the minrank; however, one way to compare the speed of these algorithms is to compute the size of the search space that would be required to compute the minrank exactly.  As we describe below, the search space for Algorithm~\ref{alg:2xdecentr} is much smaller than that for the other algorithms.  (We note that if the minrank is computed exactly, then Combined LT-CMAR approach of previous work becomes an exact algorithm, while our algorithm is a two-approximation.)

In~\cite{kim2019linear}, the set of possible matrices are those that fit a constraint matrix $C$: a matrix $M$ of the same dimensions as $F$ is a possible solution if $F_{ij}=0 \to M_{ij} = 0$. The number of linearly independent rows of such a solution $M$ is the corresponding solution size, so the goal is to find the $M$ of minimum rank. Let $n_k=|B_k|_1$, i.e. the number of blocks node $k$ has. In this approach,  the search space of matrices that fit the requirements to decode is of size
\[ 2^{\sum_{k=1}^n \sum_{k'=1}^n |R_k \odot (B_k'+R_k')|_1}\leq 2^{\sum_{k=1}^n n\cdot n_k}. \]

In~\cite{li2018multi}, the minimization problem is over a smaller search space of matrices, but with additional constraints. As in other work, a solution matrix $M$ represents the requested blocks and side information of each node and the goal is to minimize $\rk_2(M)$. Additionally, $M$ must be in the rowspan of a matrix $C$, which represents the blocks available at each node to use in constructing messages. Using Gaussian elimination, the submatrices of $C$ corresponding to each sender node are altered to maximize repeated rows in $C$ as a whole. Then letting $n^{(c)}_k$ be the number of such redundant rows, the search space for solution matrices is of size
$$2^{\sum_{k=1}^n (n_k^2+n_k)/2-\sum_{k=1}^n( (n_k^{(c)})^2+n_k^{(c)})/2}.$$ 

A heuristic method, the Combined LT-CMAR procedure, for finding some of such redundancies is provided in~\cite{li2018multi}. 

In contrast,
Algorithm~\ref{alg:2xdecentr}, in particular step~\ref{alg2x:minrkstep}, requires searching over $2^{\sum_{k=1}^n n_k}$ matrices. 

Figure~\ref{fig:searchratio} shows how the base-2 logarithm of the search space of Algorithm~\ref{alg:2xdecentr} compares to that of the Combined LT-CMAR procedure of~\cite{li2018multi}. The larger the ratio, the more costly the Combined LT-CMAR algorithm is relative  to our EIC heuristic.  Values are computed on Erd\H{o}s-Renyi graphs  randomly generated with various values of $n$, the number of vertices, and $p$, the probability of each directed edge existing. Note that graphs are re-sampled for each trial until one is generated such that every node has an out-degree of at least one. This is done because a node without an out-edge cannot satisfy any requirements with messages from the other nodes, so it has to be dropped from the problem, reducing $n$. Except for the smallest values of $n$ and $p$, $S_{\model}$ is smaller than $S_{LT-CMAR}$, meaning that Algorithm~\ref{alg:2xdecentr} has a smaller search space than the combined LT-CMAR algorithm. 

As shown in Figure~\ref{fig:searchratiooverp}, the ratios go down in some cases as the edge probability approaches $1$, because denser graphs create more similarities in the neighborhoods of nodes for the Combined LT-CMAR procedure to leverage into search space reduction. However, as shown in Figure~\ref{fig:searchratioovern}, for fixed edge probabilities not close to $1$, the logarithm of the search space for our algorithm grows relative to the logarithm of the search space for Combined LT-CMAR.

\begin{figure}
	\centering  
	\subfloat[{\label{fig:searchratioovern}}]{\includegraphics[width=.24\textwidth]{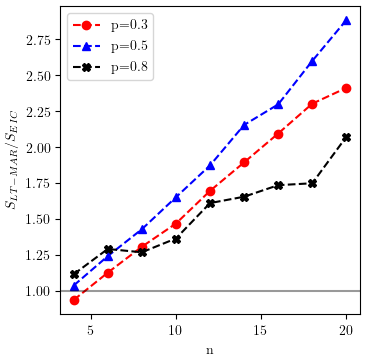}}
		\subfloat[{\label{fig:searchratiooverp}}]{\includegraphics[width=.24\textwidth]{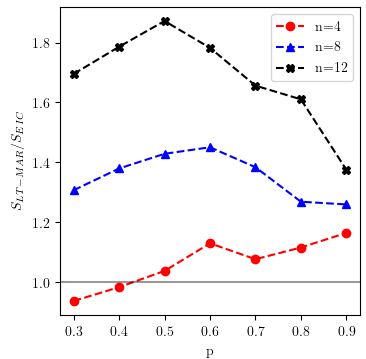}}
	\caption{$S_{LT-CMAR}$ is the value $\log_2$ of  the search space size  for a given $\probset$  using the method of~\cite{li2018multi}. $S_{\model}$  is the value $\log_2$ of the search space size using our method of approximating with the corresponding centralized solution. Ratios $S_{LT-CMAR}/S_{\model}$ are plotted for  the averages over sets of  $20$ Erd\H{o}s-Renyi graphs with the given number of vertices $n$ and probability $p$  for each directed edge. When $S_{LT-CMAR}/S_{\model}>1$ our algorithm has a strictly smaller search space. }\label{fig:searchratio}
\end{figure}


\subsection{Approximating $\T$}\label{sec:tapprox}

Computing a task-based solution consists of two main steps: finding a neighborhood partition (Definition~\ref{def:nborpartition}) and finding an index coding solution to the task defined by each $\N_i$ for sender node  $i$. 
Our heuristic uses Lemma~\ref{lem:nborupper} to approximate an
optimal choice for a neighborhood partition.  In order to see how, we define the \em neighborhood-cliques \em associated with an EIC problem:
\begin{definition}
Given a problem graph $G$ for some EIC problem $\probset$ with sender node neighborhoods $N_1,...,N_n$, let the set of \emph{neighborhood-cliques} be $\sC := \{V(C): C\text{ is a  maximal clique in }G|_{N_i} \text{ for some }N_i\}$
\end{definition}
We first define the min-cover  and min-clique-cover  problems. Let $U = \{u_1,...,u_n\}$ be a set of $n$ elements. Let $S = \{X_1,...,X_m$\} be subsets of $U$, i.e. $X_j \subseteq U$ for each $j \in[m]$, such that for all $i \in [n]$, $u_i$ is in some $X_j$. The min-cover problem over $(U,S)$ is to find the smallest $\{X_{j_1},...,X_{j_k}\}\subseteq S$, such that $\bigcup_{\ell = 1}^kX_{j_\ell} = U$. The min-clique-cover problem over a graph $G$ is an instance $(U,S)$ of min-cover in which $U=V(G)$ and $S$ is the set of all cliques in $G$, including non-maximal cliques and single vertices.

Using neighborhood-cliques, solving for the chromatic numbers used to upper bound the length of a task-based solution reduces to min-cover: 
\begin{theorem}\label{thm:mincolorreduction}
Given an EIC problem $\probset$ and the corresponding problem graph $G$, solving for the neighborhood partition $\N_1,...\N_n$ to minimize $\sum_{i=1}^n \chi(\overline{G}|_{\tilde{N}_i})$ is exactly equivalent to the min-cover problem over vertices of $G$ with sets $\sC=\{V(C_i): C_i \text{ is a  maximal clique in }G|_{N_i}\}$.
\end{theorem}

\begin{lemma}\label{lem:nsetstocliques}
Given the neighborhood partition $\{\N_i: i \in [n]\}$ for some $\probset$ with problem graph $G$, there exists a cover $\cC$ of $V(G)$ chosen from elements of $\sC$ (the set of neighborhood-cliques) such that 
\[ \sum_{C_j \in \cC} \chi(\overline{C_j}) = \sum_{i=1}^n \chi (\overline{G}|_{\N_i}).\]
\end{lemma}
\begin{proof}
Take some $\overline{\N_i}$ and consider a minimum coloring using $\chi(\overline{G}|_{\N_i})$ colors. Each set of vertices with a shared color is by definition a clique in $\Gi$, call such a clique  $C_j$. Since $\N_i\subseteq N_i$, $C_j$ (or a larger clique containing $C_j$ if $C_j$ is not maximal) is in $\sC$. We can apply this to all $\N_i$ to get a cover of $V(G)$. Since we create such a $C_j$ for each color used for $\N_i$, and the complement of a clique is 1-colorable, 
\[ \sum_{j: C_j \text{used for }\N_i}\chi(\overline{C_j})=|\{j: \text{used for }\N_i\}| = \chi(\overline{G}|_{\N_i}). \]

\end{proof}

\begin{lemma}\label{lem:cliquestonsets}
Given an EIC problem $\probset$ with problem graph $G$ and  clique cover $\cC := \{C_j\} \subseteq \sC$, there  is a corresponding choice of neighborhood partition $\N_1,...,\N_n$ such  that 
\[ \sum\chi(\overline{C_j}) \geq \sum \chi(\overline {G}|_{\tilde{N}_i}). \]
\end{lemma}
\begin{proof}
For each $i$, let $\N_i:=\bigcup \{C_j \in \cC: C_j \subseteq N_i\}$. Now letting $c :=|\{C_j \in \cC: C_j \subseteq N_i\}|$, we can color $\overline{G}|_{\N_i}$ with $c$ colors, so $\chi(\overline{G}|_{\N_i}) \leq c = \sum \chi(\overline{C_j})$.
\end{proof}

The proof of Theorem~\ref{thm:mincolorreduction} follows  immediately from Lemmas~\ref{lem:nsetstocliques} and~\ref{lem:cliquestonsets}. This also gives us an algorithm for $\N_1,...,\N_n$. Since min-cover is NP-hard solving for these will be as well,  but we can use existing min-cover approximation algorithms. Below, Algorithm~\ref{alg:tbapprox} computes the neighborhood partition which minimizes $\sum_{i=1}^n \chi(\overline{G}|_{\N_i})$ and the length of the minimum task-based solution using that partition.
\begin{algo}\label{alg:tbapprox}
Given an $EIC$ problem $\probset$:
\begin{enumerate}
\item Construct the problem graph $G$
\item Let $\sC = \emptyset$
\item For each $v_i\in V(G)$
\begin{enumerate}
\item Let $N_i\subset V(G)$ be the out-neighborhood of $v_i$
\item Compute $G|_{N_i}$, the subgraph induced by $N_i$
\item Compute the set of maximal cliques in $G|_{N_i}$ and add each to $\sC$
\end{enumerate}
\item Compute min clique cover $\cC$ of $\sC$
\item Let $T = 0$
\item For each $v_i\in V(G)$
\begin{enumerate}
\item Let $\N_i = \bigcup \{C_j \in \cC: C_j \subseteq N_i\}$
\item Compute $\Gi$, the subgraph induced by $\N_i$
\begin{enumerate}
\item Let $(R^{(i)},B^{(i)})$ be a EIC problem with problem graph $\Gi$, keeping vertex labels from $\probset$
\end{enumerate}
\item $T = T+ \rminrk_2(\Gi,\mathcal{A}_{(R^{(i)},B^{(i)})})$
\end{enumerate}
\item $T$ is the total cost of the optimal task-based solution given neighborhood partition $\{\N_1,...,\N_n\}$.
\end{enumerate}
\end{algo}

Figure~\ref{fig:costratio} shows the ratio of the length of our approximately optimal task based solution compared to the length of the optimal centralized solution. 
This ratio upper bounds the ratio of a true optimal task based solution to the corresponding centralized solution. In all of our experiments this approximation ratio is upper-bounded by $1.4$. As in the experiments in Figure~\ref{fig:searchratio}, Erd\H{o}s-Renyi graphs are randomly generated for a variety of values for $n$, the number of nodes, and $p$, the directed edge probability. As the size of the graph increases  for a fixed edge probability, the ratio appears to converge. For a fixed number of nodes, there also appears to be some upper bound on the ratio even as the probability of each edge goes to $1$.

\begin{figure}
	\centering  
	\subfloat[{\label{fig:costratioovern}}]{\includegraphics[width=.24\textwidth]{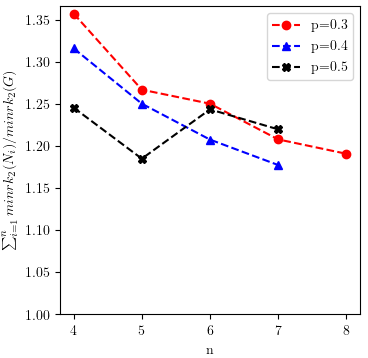}}
		\subfloat[{\label{fig:costratiooverp}}]{\includegraphics[width=.24\textwidth]{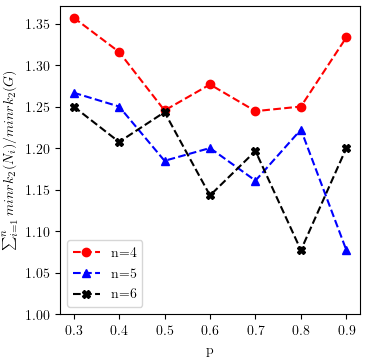}}
\caption{
Ratio of the length of our task-based solution returned by Algorithm~\ref{alg:tbapprox} to the length of the optimal centralized solution.\protect\footnotemark}
\label{fig:costratio}
\end{figure}

\footnotetext{Sample sizes in these experiments are $10$ random graphs, except $p=0.9, n=6$ which only uses $5$, since the search space for the brute-force minrank algorithm explodes, increasing exponentially in the number of graph edges.}

\section{Conclusion}
In this paper we  defined embedded index coding, a special case of multi-sender index coding in which each node of the network is both a broadcast sender and a  receiver.  We characterized an EIC problem using a problem graph, and we used this formulation to show that the optimal length of a solution to an EIC problem is bounded by twice the length of the optimal centralized index coding solution.  We also defined \em task-based \em solutions to EIC problems, in which the set messages broadcast by each node can be decoded independently of messages from other senders, and we proved characterizations and bounds for task-based solutions.  Finally, we used these bounds to develop heuristics for finding good solutions to EIC problems, and showed empirically that these heuristics perform well.

We end with some open questions and future directions.
Since this work first appeared, 
it was shown by~\cite{haviv2019taskbased} that for any integer $k$, there exists an index coding problem with problem graph $G$ and $\minrk_2(G) = k$, such that the task-based solution cost is $\Theta(k^2)$. Since we've shown a decentralized solution has cost within a constant factor of the centralized solution cost, i.e. $\minrk_2(G)$, this result also shows a gap between general decentralized and task-based solutions. However, the exact relationship between decentralized solutions and centralized solutions to embedded problems remains open.
 
It is also an interesting question to 
improve on algorithms for finding task-based solutions.  Our current approach uses an upper bound on minrank, given by the chromatic number of the complement of the problem graph.  This bound is known to be quite loose in some settings. The fractional chromatic number of the complement of the problem graph, $\chi_f(\overline{G})$, has been used to tighten the upper on minrank of $G$~\cite{blasiak2010index}, and it was also shown by~\cite{natarajan2018locally} that the optimal centralized index coding solution size with locality of one is $\chi_f(\overline{G})$. Thus the fractional chromatic number may be a useful approach in this direction. 



\bibliographystyle{plain}
\bibliography{citations}

\end{document}